\documentclass[12pt]{article}
\usepackage{appendix}
\usepackage{extarrows}
\usepackage{multirow}%
\usepackage{amsmath,amssymb,amsfonts}%
\usepackage{enumitem}
\usepackage{amsthm}%
\usepackage{mathrsfs}%
\allowdisplaybreaks[4]
\usepackage{tikz}%
\usepackage{subcaption}
\usepackage[numbers,sort&compress,comma,square]{natbib}
\usepackage{booktabs} % 导入booktabs宏包
\usepackage{hyperref}

\hypersetup{
	colorlinks=true,
	anchorcolor=yellow,
	linkcolor=cyan,
	filecolor=blue,
	urlcolor=cyan,
	citecolor=green,
}
\usepackage{cleveref}
\crefformat{equation}{Eq.~(#2#1#3)}

\usepackage[margin=1in]{geometry}

\theoremstyle{definition}

\theoremstyle{plain}
\newtheorem{lem}{Lemma}[section]
\newtheorem{thm}{Theorem}[section]
\newtheorem{cor}{Corollary}[section]
\newtheorem{con}{Conjecture}[section]
\newtheorem{prop}{Proposition}[section]
\theoremstyle{remark}

\newtheorem{ex}{\bf Example}

\newtheoremstyle{case}{}{}{}{}{}{:}{ }{}
\theoremstyle{case}

\DeclareMathOperator{\ii}{\mathrm{i}}
\DeclareMathOperator{\R}{\mathbb{R}}

\begin{document}

\title{\textbf{Zero transfer on mixed graphs}}

\author{
	Xingkun Song$^{1,2}$\thanks{Email: \href{mailto:xksong@126.com}{xksong@126.com}}
	\quad
	Huiqiu Lin$^{3}$\thanks{Corresponding author. Email: \href{mailto:huiqiulin@126.com}{huiqiulin@126.com}}\\[2ex]
	{\small $^{1}$ School of Mathematics and Statistics, Qinghai Minzu University,}\\
	{\small Xining, Qinghai 810007, P.R. China}\\[1ex]
	{\small $^{2}$ Qinghai Institute of Applied Mathematics,}\\
	{\small Xining, Qinghai 810007, P.R. China}\\[1ex]
	{\small $^{3}$ School of Mathematics, East China University of Science and Technology,}\\
	{\small Shanghai 200237, P.R. China}
}

\date{}
\maketitle

\begin{abstract}
	In this paper, we investigate zero transfer on mixed graphs. Zero transfer is a quantum walk phenomenon in which the transition amplitude between two vertices is identically zero for all times, so that no quantum state transfer occurs between them. Using the Hermitian adjacency matrix, we derive necessary and sufficient conditions for zero transfer in mixed graphs. We then specialize these criteria to oriented circulant graphs, obtaining nonexistence results for prime order, structural restrictions for even order, and exhaustive computational classifications for small orders.
\end{abstract}

\begin{flushleft}
\end{flushleft}

\section{Introduction}

For over two decades, quantum state transfer on graphs has been an important topic in quantum information science and quantum computation. In this context, a network of interacting particles is modeled as a graph, where vertices represent quantum states (or qubits) and edges denote their interactions. The evolution of the system is governed by the continuous-time Schr\"{o}dinger equation, which gives rise to a continuous-time quantum walk. Such quantum walks are usually generated by the adjacency matrix, the Laplacian matrix, or other matrices associated with graphs.

One important concept in quantum information is perfect state transfer, introduced by Bose in 2003~\cite{Bo03}, which describes the transfer of quantum states between vertices in a graph. A graph is said to have perfect state transfer from vertex $u$ to vertex $v$ if there exists a time $t$ at which the transition probability from $u$ to $v$ is exactly one. In undirected graphs, because of the symmetry of the underlying Hamiltonian, perfect state transfer typically occurs simultaneously in both directions, that is from $u$ to $v$ and from $v$ to $u$.

However, for oriented graphs, where edges have directions, this symmetry is broken. To study quantum walks on such graphs, researchers have proposed the Hermitian adjacency matrix, which remains Hermitian while encoding edge orientations~\cite{CFGHST14,CGKST17,CCC23,S22,SL22}. This asymmetry leads to phenomena that do not occur in undirected graphs. For instance, perfect state transfer may occur only in one direction, which is referred to as one-way perfect state transfer. Another related phenomenon is multiple state transfer, in which transfer may involve more than two vertices. In contrast, zero transfer means that a quantum state never transfers from one vertex to another at any time~\cite{SPFW19}. Beyond being a theoretical counterpart to perfect state transfer, zero transfer has potential relevance in quantum architecture design. It can be used to model quantum shielding or path isolation, thereby preventing the leakage of quantum information into unwanted regions of a quantum network.

The concept of zero transfer was introduced by Sett, Pan, Falloon, and Wang in 2019~\cite{SPFW19}. It describes the situation in which the transition probability from vertex $u$ to vertex $v$ is always zero. While perfect state transfer has been widely studied (see~\cite{CG21b,Go08,Go12a,Go12b}), zero transfer remains less developed. Some recent works have obtained partial results on this problem. For example, Coutinho and Godsil~\cite{CG21b} characterized infinite families with zero transfer using signed adjacency matrices and coalescence (1-sums) of two graphs. More recently, Chaves, Chagas, and Coutinho~\cite{CCC23} showed that zero transfer can occur in cycles with weighted adjacency matrices.

The primary goal of this paper is to characterize zero transfer in mixed graphs using the Hermitian adjacency matrix. The remainder of this paper is organized as follows. In Section~\ref{sec::4}, we review the necessary definitions and preliminaries. In Section~\ref{sec::2}, we establish necessary and sufficient conditions for zero transfer in mixed graphs. In Section~\ref{sec::3}, we investigate zero transfer in oriented circulant graphs, prove that it does not occur when the order of the graph is prime, and establish structural and parity restrictions for even order. To complement the analytical results, we use the open-source software SageMath~\cite{sage} to compute and classify all instances of zero transfer in oriented circulant graphs of order up to 20 (see~\Cref{Appendix}).

\section{Definitions and Preliminaries}\label{sec::4}

A \textbf{mixed graph} $\Gamma=(V, E, A)$ consists of a set of vertices $V$, a set of undirected edges $E$, and a set of directed edges (or arcs) $A$.  Specifically, $\Gamma$ is \textbf{undirected} (resp. \textbf{oriented}) if it contains only undirected (resp. directed) edges.

The \textbf{Hermitian adjacency matrix} of $\Gamma$, introduced independently by Liu and Li~\cite{LL15} and by Guo and Mohar~\cite{GM17}, is the complex matrix $H_\Gamma = (h_{uv})_{u,v \in V}$ defined by
\begin{equation*}
	h_{uv} =
	\begin{cases}
		1,    & \text{ if }	\{u,v\}\in E, \\
		\ii,  & \text{ if } (u,v)\in A,   \\
		-\ii, & \text{ if } (v,u)\in A,   \\
		0,    & \text{ otherwise},
	\end{cases}
\end{equation*}
where $\ii=\sqrt{-1}$. Clearly, $H_\Gamma$ is a Hermitian matrix. In particular, if $\Gamma$ is undirected (resp. oriented), then $H_\Gamma=A_\Gamma$ (resp. $H_\Gamma=\ii S_\Gamma$), where $A_\Gamma$ (resp. $S_\Gamma$) is the \textbf{adjacency matrix} (resp. \textbf{skew adjacency matrix}) of $\Gamma$. The eigenvalues of $H_\Gamma$ are also referred to as the \textbf{eigenvalues} of $\Gamma$.

Let $\Gamma=(V, E, A)$ be a mixed graph with Hermitian adjacency matrix $H_{\Gamma}$. We define the  \textbf{transition matrix} of $H_\Gamma$ by
\[
	U(t)=\exp(-\ii t H_\Gamma),
\]
where $\ii=\sqrt{-1}$ and $t$ is a real number. Note that $U(t)$ is unitary for all $t \in \R$, and satisfies the following fundamental properties:
\begin{align}
	\text{\bfseries P1: } & U(-t)=U(t)^{-1}=U(t)^{*}, \label{eq::1a} \\
	\text{\bfseries P2: } & U(t+t')=U(t)U(t'),  \label{eq::1b}
\end{align}
where $\cdot^*$ denotes the conjugate transpose.

For any vertex $v\in V$, let $\mathbf{e}_v$ be the vector defined on $V$ such that $\mathbf{e}_v(u)=1$ if $u=v$ and $\mathbf{e}_v(u)=0$ otherwise. For two vertices $u$ and $v$ of $\Gamma$, we say that $\Gamma$ has \textbf{perfect state transfer} from $u$ to $v$ if there exists a time $t \in \R$ such that
\begin{equation*}
	U(t)\mathbf{e}_u= \gamma \mathbf{e}_v.
\end{equation*}
Here, $\gamma$ is a complex number of modulus 1.

In contrast, we say that $\Gamma$ has \textbf{zero transfer} from $u$ to $v$
if
\[
	U(t)_{u,v} = 0.
\]
for any time $t\geq 0$. Since $U(t)$ is unitary, it follows that $U(t)_{v,u} = 0$.
Therefore, we simply say that there exists zero transfer between vertices $u$ and $v$.

\section{Characterization of zero transfer}\label{sec::2}

Let $\Gamma$ be a mixed graph with Hermitian adjacency matrix $H$. The Taylor expansion of the transition matrix is
\begin{equation}\label{eq:2}
	U(t) = \exp(-\ii t H) = \sum_{k \geqslant 0} \frac{ (-\ii t H)^{k} }{k!}.
\end{equation}

\begin{lem}\label{thm:3}
	Let $\Gamma$ be a mixed graph with Hermitian adjacency matrix $H$, and let $u$ and $v$ be two distinct vertices of $\Gamma$. Then there is zero transfer between $u$ and $v$ if and only if $(H^k)_{u,v}=0$ for every positive integer $k$.
\end{lem}

\begin{proof}
	By~\Cref{eq:2}, the $(u,v)$-entry of $U(t)$ is the analytic function
	\[
		U(t)_{u,v}=\sum_{k\geq 0}\frac{(-\ii t)^k}{k!}(H^k)_{u,v}.
	\]
	Since $u\neq v$, the term with $k=0$ is zero. Thus $U(t)_{u,v}$ is identically zero if and only if all its Taylor coefficients at $t=0$ vanish, which is equivalent to $(H^k)_{u,v}=0$ for every $k\geq 1$.

\end{proof}

To further simplify the criterion, we use the minimal polynomial of $H$.
Recall that the \textbf{minimal polynomial} $m_H(x)$ of a Hermitian matrix $H$
is the unique monic polynomial of least degree $s$ such that $m_H(H)=0$.
By the {\bf division algorithm} for polynomials, for any positive integer $k$
there exist a polynomial $q(x)$ and a remainder polynomial $r(x)$ with $\deg(r(x))<s$
such that
\begin{equation*}
	H^k = m_H(H) q(H) + r(H)=r(H).
\end{equation*}
Consequently, each power $H^k$ can be written as a linear combination of
$I, H, H^2, \ldots, H^{s-1}$.

Since Hermitian matrices are diagonalizable, the matrix exponential can be expressed explicitly in terms of the matrix powers using the {\bf Lagrange interpolation formula}.
In particular, Ben Taher and Rachidi~\cite{BR02} provided such an expression for a general square matrix $A$.
We adapt their result to the case of Hermitian matrices as follows.

\begin{lem}{\upshape\bfseries (\cite[Corollary 2]{BR02})}\label{lem:1}
	Let $H$ be a Hermitian matrix and suppose that its minimal polynomial is
	$m_H(x)=\prod_{j=1}^s (x-\theta_j)$, where $\theta_1,\ldots,\theta_s$ are the distinct eigenvalues of $H$. Then
	$$
		\exp(-\ii t H)=\sum_{j=1}^s \exp(-\ii t \theta_j)
		\prod_{\substack{d=1 \\ d \neq j}}^s \frac{H-\theta_d I}{\theta_j-\theta_d} .
	$$
\end{lem}

It follows from~\Cref{lem:1} that $\exp(-\ii t H)$ can be written as a linear combination of
$I, H, \ldots, H^{s-1}$. By~\Cref{thm:3}, we obtain the following theorem.

\begin{thm}\label{lem::3}
	Let $\Gamma$ be a mixed graph with Hermitian adjacency matrix $H$, and let $u$ and $v$ be two distinct vertices. Let $m_H(x)$ be the minimal polynomial of $H$, and let $s=\deg m_H$. Then there is zero transfer between $u$ and $v$ if and only if $ (H^k)_{u,v} = 0$ for all integers $k$ with $1\leq  k \leq s-1$.
\end{thm}

From~\Cref{lem::3}, we immediately obtain the following properties of zero transfer.

\begin{prop}%\label{lem::4}
	Let $\Gamma$ be a mixed graph with Hermitian adjacency matrix $H$, and let $u$ and $v$ be two distinct vertices. If there exists zero transfer between vertices $u$ and $v$, then the following properties hold:
	\begin{enumerate}[label = \bf(\roman*)]
		\item vertices $u$ and $v$ are not adjacent,
		\item for every positive integer $k$, the weighted sum of all walks of length $k$ from $u$ to $v$ is zero. In particular, the number of such walks is even.
	\end{enumerate}
\end{prop}

\begin{proof}
	By~\Cref{thm:3}, zero transfer implies $(H^k)_{u,v}=0$ for every $k\geq1$. Taking $k=1$ gives $H_{u,v}=0$, so $u$ and $v$ are not adjacent. For general $k$, the entry $(H^k)_{u,v}$ is the sum, over all walks of length $k$ from $u$ to $v$, of the products of the corresponding edge weights. Hence this weighted sum is zero. Since each walk weight belongs to $\{1,-1,\ii,-\ii\}$, a zero sum forces the numbers of weights $1$ and $-1$ to be equal and the numbers of weights $\ii$ and $-\ii$ to be equal; therefore the total number of walks is even.
\end{proof}

The Hermitian adjacency matrix is a normal matrix, which means it admits a spectral decomposition, as described in Godsil and Royle~\cite{GR01}. Specifically, let $H$ be the Hermitian adjacency matrix of a mixed graph with distinct eigenvalues $\theta_0, \theta_1, \ldots, \theta_d$. Then $H$ can be written as

\begin{equation*}
	H = \sum_{r=0}^{d} \theta_{r} E_{\theta_{r}},
\end{equation*}
where $E_{\theta_{r}}$ denotes the \textbf{spectral idempotent} (or \textbf{eigenprojector}) corresponding to the eigenvalue $\theta_r$. The spectral idempotents $E_{\theta_{r}}$ satisfy the following standard properties:

\begin{enumerate}[label = \bf(\roman*)]
	\item $E_{\theta_{i}} E_{\theta_{j}} = 0$, for $i\neq j$;
	\item $E_{\theta_{i}}^2 = E_{\theta_{i}}$, for any $i$;
	\item $\sum_{i=0}^{d}E_{\theta_{i}}=I$;
\end{enumerate}
where $I$ denotes the identity matrix.

By applying the spectral decomposition of $H$, the corresponding \textbf{transition matrix} $U_H(t)$ can be expressed as
\begin{equation}\label{eq::2}
	U_{H} (t) =\exp(-\ii t H)= \sum_{r=0}^{d} e^{-\ii t \theta_r} E_{\theta_{r}},
\end{equation}
where $\theta_r$ is the $r$-th eigenvalue of $H$, and $E_{\theta_{r}}$ is the corresponding spectral idempotent.

For a vertex $a \in V$, the \textbf{eigenvalue support} at $a$ is defined as the set
$$\Phi_a=\{\theta_r: E_{\theta_{r}} \mathbf{e}_a \neq \mathbf{0}\}.$$
According to spectral decomposition, we easily obtain the following theorem.

\begin{thm}\label{thm:2}
	Let $\Gamma$ be a mixed graph with Hermitian adjacency matrix $H$, and let $u$ and $v$ be two distinct vertices. Then the following statements are equivalent:
	\begin{enumerate}[label = \bf(\roman*)]
		\item There is zero transfer between vertices $u$ and $v$.
		\item For every eigenvalue $\theta_r$, $(E_{\theta_{r}})_{u,v}=0$. Equivalently, either $\Phi_u \cap \Phi_v = \emptyset$, or $(E_{\theta_{r}})_{u,v}=0$ for all $\theta_r \in \Phi_u \cap \Phi_v$.
	\end{enumerate}
\end{thm}
\begin{proof}
	If $(E_{\theta_r})_{u,v}=0$ for every $r$, then~\Cref{eq::2} gives $U_H(t)_{u,v}=0$ for all $t$, so zero transfer occurs.

	Conversely, suppose that zero transfer occurs. By~\Cref{eq::2},
	\[
		\sum_{r=0}^{d} e^{-\ii t\theta_r}(E_{\theta_r})_{u,v}=0
		\qquad\text{for all }t\in\mathbb R.
	\]
	Let $x_r=(E_{\theta_r})_{u,v}$. Differentiating this identity $k=0,1,\ldots,d$ times and setting $t=0$ gives
	\[
		\sum_{r=0}^{d}\theta_r^k x_r=0,\qquad k=0,1,\ldots,d.
	\]
	The coefficient matrix is a Vandermonde matrix in the distinct eigenvalues $\theta_0,\ldots,\theta_d$, and hence is nonsingular. Therefore $x_r=0$ for every $r$. The equivalent formulation using $\Phi_u\cap\Phi_v$ follows because $(E_{\theta_r})_{u,v}=0$ whenever $\theta_r$ is not in the support of at least one of $u$ and $v$.
\end{proof}

Thus zero transfer can occur either because the relevant eigenvalue supports do not meet, as happens for vertices in different connected components, or because every common spectral idempotent has zero $(u,v)$-entry.

Let $\phi(G,x)$ denote the characteristic polynomial of the Hermitian adjacency matrix of $G$. For vertices $a,b\in V(G)$, let $G\setminus a$, $G\setminus b$, and $G\setminus ab$ be the subgraphs obtained by deleting the vertices $a$, $b$, and $\{a,b\}$, respectively.
Let $W_{ab}(G,x)$ be the \textbf{walk generating function} of $G$, which counts weighted walks from $a$ to $b$, where the weight of a walk is the product of the weights of its edges~\cite{G93}. Then
\begin{equation}\label{eq::a}
	x^{-1}W_{ab}(G,x^{-1}) \;=\; \frac{\big(\phi(G\setminus a, x)\,\phi(G\setminus b, x)\;-\;\phi(G\setminus ab, x)\,\phi(G, x)\big)^{1/2}}{\phi(G, x)} .
\end{equation}

According to Coutinho and Godsil~\cite[pp.~67, 76]{CG21b}, we have
\begin{equation}\label{eq::b}
	(E_r)_{a,b}
	\;=\;
	\sum_{P\in \mathcal{P}} \mathrm{wt}(P)\,
	\frac{\phi(G\setminus P, x)(x-\theta_r)}{\phi(G, x)}\Bigg|_{x=\theta_r}
	\;=\;
	x^{-1}W_{ab}(G,x^{-1})(x-\theta_r)\Big|_{x=\theta_r},
\end{equation}
where $\mathcal{P}$ denotes the set of paths from $a$ to $b$, and $\mathrm{wt}(P)$ denotes the weight of the path $P$. By~\Cref{thm:2},~\Cref{eq::a}, and \Cref{eq::b}, we have the following proposition.

\begin{prop}\label{prop:zerotransfer}
	Let $a,b\in V(G)$. Then there is zero transfer between vertices $a$ and $b$ if and only if
	\[
		\phi(G\setminus a, x)\,\phi(G\setminus b, x)
		\;=\;
		\phi(G\setminus ab, x)\,\phi(G, x).
	\]
\end{prop}

\begin{proof}
	By~\Cref{thm:2}, zero transfer between $a$ and $b$ is equivalent to $(E_r)_{a,b}=0$ for every spectral idempotent. In view of~\Cref{eq::b}, this is equivalent to the vanishing of the walk generating function $W_{ab}$. Formula~\Cref{eq::a} then gives the stated polynomial identity, and the converse follows by reversing the same implications.
\end{proof}

\begin{lem}{\bf (Perron--Frobenius Theorem)}\label{thm:1}
	If $A$ is an irreducible nonnegative matrix of order $n$ with $n>2$, then the following statements hold.
	\begin{enumerate}[label = \bf(\roman*)]
		\item $\rho(A) > 0$, and $\rho(A)$ is a simple eigenvalue of $A$.
		\item $A$ has a positive eigenvector corresponding to $\rho(A)$.
		\item All nonnegative eigenvectors of $A$ correspond to the eigenvalue $\rho(A)$.
	\end{enumerate}
\end{lem}

The Perron--Frobenius theorem states that any nonnegative irreducible matrix has a unique largest eigenvalue with a strictly positive eigenvector. For quantum walks, this implies that a connected undirected graph, whose adjacency matrix is nonnegative and irreducible, cannot have zero entries in its Perron spectral idempotent. Thus, by~\Cref{thm:2}, zero transfer is impossible in connected undirected graphs.

\begin{thm}\label{thm:nonnegative}
	Let $\Gamma$ be a graph represented by a nonnegative Hermitian matrix $H$. Then zero transfer occurs between two distinct vertices $u$ and $v$ if and only if $u$ and $v$ lie in different connected components. In particular, a connected undirected graph has no zero transfer.
\end{thm}

\begin{proof}
	If $u$ and $v$ lie in different connected components, then $H$ is block diagonal with respect to the component decomposition, and so is $U(t)=\exp(-\ii tH)$. Hence $U(t)_{u,v}=0$ for all $t$.

	Conversely, suppose that $u$ and $v$ lie in the same connected component. The corresponding block of $H$ is a nonnegative irreducible matrix. By the Perron--Frobenius theorem, its spectral radius $\rho$ is a simple eigenvalue with a strictly positive eigenvector $\mathbf{x}$. The spectral idempotent for $\rho$ is
	\[
		E_{\rho}=\frac{\mathbf{x}\mathbf{x}^{*}}{\mathbf{x}^{*}\mathbf{x}},
	\]
	and therefore $(E_\rho)_{u,v}>0$. By~\Cref{thm:2}, zero transfer cannot occur between $u$ and $v$.
\end{proof}

This result is related to the concept of \textbf{ergodicity} in \textbf{Markov chains}, where irreducibility ensures that all states are accessible. The nonnegative case is therefore rigid: zero transfer can occur only for vertices separated by components. In contrast, signed, oriented, and mixed graphs need not satisfy the nonnegativity condition in the Perron--Frobenius theorem. For example, Coutinho and Godsil~\cite{CG21b} showed that the signed cycle $C_4^-$ admits zero transfer, and extended the construction to other graphs using 1-sums (see~\cite[Section~4.10]{CG21b}).

Oriented graphs fall outside this nonnegative case because their Hermitian adjacency matrices have entries in $\{0,\ii,-\ii\}$. Phase interference among walks may therefore force entries of spectral idempotents to vanish. Motivated by this phenomenon and by the symmetry of cyclic structures, we focus on oriented circulant graphs in the following section.

\section{Zero transfer on oriented circulant graphs}\label{sec::3}

In this section, we investigate zero transfer on oriented circulant graphs. Throughout the section, all oriented circulant graphs under consideration are assumed to be connected.

Recall that a circulant graph is a Cayley graph over a cyclic group. Let $\mathbb{Z}_n$ be the additive group of integers modulo $n$, and let $\mathcal{C}$ be a subset of $\mathbb{Z}_n\setminus \{0\}$. The \textbf{circulant graph} $G(\mathbb{Z}_n,\mathcal{C})$ is defined to have vertex set $\mathbb{Z}_n$ and arc set
$$
	A=\{(a,b): b-a \in \mathcal{C},\; a,b\in \mathbb{Z}_n\}.
$$
The set $\mathcal{C}$ is called the \textbf{connection set} of $G(\mathbb{Z}_n,\mathcal{C})$. In particular, if $\mathcal{C}=\mathcal{C}^{-1}$ then $G(\mathbb{Z}_n,\mathcal{C})$ is \textbf{undirected circulant graphs}, while if $\mathcal{C}\cap \mathcal{C}^{-1}=\emptyset$ then $G(\mathbb{Z}_n,\mathcal{C})$ is \textbf{oriented circulant graphs}.

Let $\Gamma=G(\mathbb{Z}_n,\mathcal{C})$ be an oriented circulant graph, and let $H$ be the  Hermitian adjacency matrix of $\Gamma$. According to~\cite{MB21a}, the eigenvalues and their corresponding eigenvectors of $H$ are given by
\begin{equation}\label{eq::1}
	\mu_j=\ii \sum_{k \in \mathcal{C}} (\omega^{jk}_n-\omega^{-jk}_n),
	\quad \mathbf{v}_j=[1 \ \omega_n^j \ \omega_n^{2j} \cdots
			\omega_n^{(n-1)j}]^{\top},
\end{equation}
for $j=0,1,\ldots,n-1 $, where $\omega_n=\exp(2\pi\ii /n)$ is a primitive $n$-th root of unity. Using $\omega_n^{jk}-\omega_n^{-jk}=2\ii \sin\left(\frac{2\pi jk}{n}\right)$, we obtain
\begin{equation}\label{eq::22}
	\mu_j=-2 \sum_{k \in \mathcal{C}} \sin\left(\frac{2\pi j k}{n}\right),
\end{equation}
for $j=0,1,\ldots,n-1 $.

The vectors $\mathbf{u}_k=\mathbf{v}_k/\sqrt{n}$ form an orthonormal eigenbasis. Hence the transition matrix $U(t)$ of $H$ can be expressed as
\begin{equation}\label{eq::4}
	U(t)=\frac1n \sum_{r=0}^{n-1} \exp(-\ii \mu_r t) \mathbf{v}_r \mathbf{v}_r^*.
\end{equation}
In particular, by \Cref{eq::1} and \cref{eq::4}, for vertices $u,v \in \mathbb{Z}_n$, we have
\begin{equation}\label{eq:7}
	U(t)_{u,v}=\frac1n \sum_{r=0}^{n-1} \exp(-\ii \mu_r t)
	\omega_n^{r(u-v)}.
\end{equation}

Let $\theta_1, \theta_2, \ldots, \theta_d$ be the distinct eigenvalues of $\Gamma$, with multiplicities $k_1, k_2, \ldots, k_d$, respectively. For each $1 \le i \le d$, let
\[
	M_{\theta_i} = \{j \in \mathbb{Z}_n \mid \mu_j = \theta_i\}
\]
be the index set corresponding to the eigenvalue $\theta_i$. It follows that $|M_{\theta_i}|=k_i$. By~\Cref{eq:7}, we have

\begin{equation*}
	U(t)_{u,v}=\sum_{i=1}^{d}\exp(-\ii t \theta_i)(E_{\theta_{i}})_{u,v}=\sum_{i=1}^{d}\exp(-\ii t \theta_i)\left(\frac1n \sum_{r\in M_{\theta_i}} \omega_n^{r(u-v)}\right).
\end{equation*}
Hence, the spectral idempotent corresponding to the eigenvalue $\theta_i$ satisfies
\begin{equation}\label{eq::6}
	(E_{\theta_{i}})_{u,v}=\frac1n \sum_{r\in M_{\theta_i}} \omega_n^{r(u-v)}.
\end{equation}

\begin{lem}\label{lem::4}
	Let $\Gamma = G(\mathbb{Z}_n,\mathcal{C})$ be an oriented circulant graph. Then the eigenvalue support $\Phi_u$ of a vertex $u$ contains all eigenvalues.
\end{lem}

\begin{proof}
	Suppose that $\theta_i\notin \Phi_u$. Then $E_{\theta_{i}} \mathbf{e}_u=\mathbf{0}$, and hence
	\[
		(E_{\theta_{i}})_{u,u}=\langle E_{\theta_{i}}\mathbf{e}_u,E_{\theta_{i}}\mathbf{e}_u\rangle=\|E_{\theta_{i}}\mathbf{e}_u\|^2=0.
	\]
	However,
	\[
		(E_{\theta_{i}})_{u,u}=\frac1n \sum_{r\in M_{\theta_i}}\omega_n^{r(u-u)}=\frac{|M_{\theta_i}|}{n}=\frac{k_i}{n}\neq0,
	\]
	which is a contradiction. Therefore $\theta_i\in\Phi_u$, and the proof is complete.
\end{proof}

By~\Cref{lem::4}, the eigenvalue support of each vertex contains all eigenvalues. Based on~\Cref{thm:2} and \Cref{eq::6}, we obtain a criterion for determining whether zero transfer occurs in oriented circulant graphs, which is stated in the following lemma.

\begin{lem}\label{lem:zero-criterion}
	Let $\Gamma = G(\mathbb{Z}_n, \mathcal{C})$ be an oriented circulant graph, and let $u$ and $v$ be two distinct vertices. Then there is zero transfer between vertices $u$ and $v$ if and only if
	\[
		(E_{\theta_{i}})_{u,v}=\frac1n \sum_{r\in M_{\theta_i}} \omega_n^{r(u-v)} = 0
	\]
	for every eigenvalue $\theta_i$.
\end{lem}

\begin{proof}
	By~\Cref{lem::4}, every eigenvalue belongs to the eigenvalue support of every vertex. Therefore~\Cref{thm:2} says that zero transfer between $u$ and $v$ occurs if and only if $(E_{\theta_i})_{u,v}=0$ for every distinct eigenvalue $\theta_i$. Formula~\Cref{eq::6} gives the stated condition.
\end{proof}

Suppose that there exists an eigenvalue of multiplicity one. Then we obtain the following lemma.

\begin{lem}\label{lem::6}
	Let $\Gamma = G(\mathbb{Z}_n, \mathcal{C})$ be an oriented circulant graph. If there exists an eigenvalue of $\Gamma$ with multiplicity one, then $\Gamma$ has no zero transfer.
\end{lem}

\begin{proof}
	If $\theta_i$ is an eigenvalue of multiplicity one, then $M_{\theta_i}=\{r\}$ for some $r\in\mathbb Z_n$.  Hence, for any vertices $u,v$,
	\[
		(E_{\theta_i})_{u,v} = \frac{1}{n}\omega_n^{r(u-v)} \neq 0.
	\]
	This contradicts the condition in \Cref{lem:zero-criterion}. Hence, $\Gamma$ has no zero transfer.
\end{proof}

\begin{lem}\label{lem::5}
	Let $\Gamma = G(\mathbb{Z}_n, \mathcal{C})$ be an oriented circulant graph. If $n$ is prime, then $\Gamma$ has eigenvalue $0$ with multiplicity one.
\end{lem}

\begin{proof}
	Let $H$ be the Hermitian adjacency matrix of $\Gamma$. From Eq.~\eqref{eq::1}, the eigenvalues $\mu_j$ of $H$ are given by
	$$ \mu_j = \ii \sum_{k \in \mathcal{C}} (\omega_n^{jk} - \omega_n^{-jk}), \quad j = 0, 1, \dots, n-1, $$
	where $\omega_n = \exp(2\pi\ii/n)$.
	Clearly, for $j=0$, we have $\mu_0 = \ii \sum_{k \in \mathcal{C}} (1 - 1) = 0$. Thus, $0$ is an eigenvalue of $H$.

	To show that the multiplicity of $0$ is exactly one, we will prove that $\mu_j \neq 0$ for all $1 \le j \le n-1$. Suppose, for the sake of contradiction, that $\mu_j = 0$ for some $j \in \{1, \dots, n-1\}$. This implies
	$$ \sum_{k \in \mathcal{C}} \omega_n^{jk} - \sum_{k \in \mathcal{C}} \omega_n^{-jk} = 0. $$

	Let $S_1 = \{jk \pmod n \mid k \in \mathcal{C}\}$ and $S_2 = \{-jk \pmod n \mid k \in \mathcal{C}\}$. Since $n$ is prime and $j \not\equiv 0 \pmod n$, the mapping $x \mapsto jx \pmod n$ is a bijection on $\mathbb{Z}_n \setminus \{0\}$. Because $\Gamma$ is an oriented graph, its connection set satisfies $\mathcal{C} \cap \mathcal{C}^{-1} = \emptyset$, which ensures that $S_1 \cap S_2 = \emptyset$.

	Thus, our assumption leads to a linear relation over $\mathbb{Q}$:
	$$ \sum_{s \in S_1} \omega_n^s - \sum_{t \in S_2} \omega_n^t = 0. $$
	Because $n$ is prime, the minimal polynomial of $\omega_n$ over $\mathbb{Q}$ is the $n$-th cyclotomic polynomial $\Phi_n(x) = 1 + x + \dots + x^{n-1}$. Consequently, any $\mathbb{Q}$-linear dependence relation among the elements $1, \omega_n, \dots, \omega_n^{n-1}$ must be a constant multiple of
	$$ 1 + \omega_n + \omega_n^2 + \dots + \omega_n^{n-1} = 0. $$
	However, since $0 \notin \mathcal{C}$, it follows that $0 \notin S_1 \cup S_2$. Therefore, the coefficient of $\omega_n^0 = 1$ in our relation is $0$. This forces the constant multiple to be $0$, meaning the entire linear relation must be trivial (all coefficients must be zero).

	This contradicts the fact that $S_1$ and $S_2$ are non-empty sets with coefficients $1$ and $-1$, respectively. Therefore, $\mu_j \neq 0$ for all $j \neq 0$, and the eigenvalue $0$ has multiplicity of exactly one.
\end{proof}

Combining~\Cref{lem::6},~\Cref{lem::5}, we obtain the following theorem.

\begin{thm}\label{thm:11}
	Let $\Gamma = G(\mathbb{Z}_n, \mathcal{C})$ be an oriented circulant graph of prime order. Then there is no zero transfer in $\Gamma$.
\end{thm}

\begin{proof}
	By~\Cref{lem::5}, $\Gamma$ has an eigenvalue with multiplicity one. By~\Cref{lem::6}, no zero transfer can occur in $\Gamma$.
\end{proof}

By~\Cref{thm:11}, zero transfer does not occur in connected oriented circulant graphs of prime order. However, this property does not hold for all odd orders. When $n$ is odd composite, the cyclotomic field $\mathbb{Q}(\omega_n)$ contains proper subfields, allowing sums of roots of unity over certain index sets to vanish.

Computational searches show that zero transfer also does not occur for odd composite orders $n < 21$ (i.e., $n=9$ and $15$). The following example shows that $n=21$ is the minimal odd order admitting zero transfer.

\begin{ex}\label{ex:counterexample_21}
	Consider the connected oriented circulant graph $\Gamma = G(\mathbb{Z}_{21}, \mathcal{C})$ with connection set $\mathcal{C} = \{2, 10, 12, 15, 16, 17\}$. Let $u, v \in \mathbb{Z}_{21}$ be two vertices with distance $d = u - v \equiv 7 \pmod{21}$. The $(u,v)$-entry of the spectral idempotent $E_{\theta_i}$ is given by
	$$ (E_{\theta_i})_{u,v} = \frac{1}{21} \sum_{r \in M_{\theta_i}} \omega_{21}^{7r} = \frac{1}{21} \sum_{r \in M_{\theta_i}} \omega_3^r, $$
	where $\omega_3 = e^{2\pi i/3}$ and $M_{\theta_i}$ is the index set associated with the eigenvalue $\theta_i$.

	The eigenvalues of $\Gamma$ partition $\mathbb{Z}_{21}$ into exactly seven index sets:
	$M_{\theta_1} = \{0, 7, 14\}$,
	$M_{\theta_2} = \{1, 8, 18\}$,
	$M_{\theta_3} = \{2, 15, 16\}$,
	$M_{\theta_4} = \{3, 13, 20\}$,
	$M_{\theta_5} = \{4, 9, 11\}$,
	$M_{\theta_6} = \{5, 6, 19\}$, and
	$M_{\theta_7} = \{10, 12, 17\}$.

	Observe that for every $1 \le i \le 7$, $M_{\theta_i} \equiv \{0, 1, 2\} \pmod 3$. Since $\omega_3^r$ depends only on $r \pmod 3$, the sum over each index set evaluates to $1 + \omega_3 + \omega_3^2 = 0$. Consequently, $(E_{\theta_i})_{u,v} = 0$ for all $i$. Therefore, zero transfer occurs between $u$ and $v$.
\end{ex}

As shown in \Cref{ex:counterexample_21}, zero transfer is closely related to index sets forming complete residue systems modulo $3$. In this family, the simultaneous vanishing of the spectral idempotents occurs only at distances that are nonzero multiples of $n/3$.

For $n=21$, a computational search shows that there exist exactly 12 connection sets admitting zero transfer. For each graph in this family, zero transfer occurs precisely between vertices $u,v \in \mathbb{Z}_{21}$ satisfying $u-v \in \{7,14\}$. The corresponding 12 connection sets (including $\mathcal{C}_7$ in \Cref{ex:counterexample_21}) are listed in~\Cref{tab:connectionsets}.

\begin{table}[htbp]
	\centering
	\caption{Connection sets admitting zero transfer for $n=21$}
	\label{tab:connectionsets}

	\begin{tabular}{cccc}
		\toprule
		$i$ & $\mathcal{C}_i$       & $i$  & $\mathcal{C}_i$        \\
		\midrule
		$1$ & $\{1,3,5,8,9,19\}$    & $7$  & $\{2,10,12,15,16,17\}$ \\
		$2$ & $\{1,5,6,8,18,19\}$   & $8$  & $\{2,12,13,16,18,20\}$ \\
		$3$ & $\{1,6,8,9,10,17\}$   & $9$  & $\{3,4,5,11,15,19\}$   \\
		$4$ & $\{1,8,10,12,17,18\}$ & $10$ & $\{3,4,9,11,13,20\}$   \\
		$5$ & $\{2,3,13,15,16,20\}$ & $11$ & $\{4,5,6,9,11,19\}$    \\
		$6$ & $\{2,6,10,16,17,18\}$ & $12$ & $\{4,11,12,13,15,20\}$ \\
		\bottomrule
	\end{tabular}

\end{table}

The preceding results have focused on zero transfer in oriented circulant graphs of odd order. We now turn to the case where $n$ is even. We first consider the case in which all eigenvalues have multiplicity at least two.

\begin{thm}
	Let $\Gamma=G(\mathbb{Z}_n,\mathcal{C})$ be an oriented circulant graph. If zero transfer occurs in $\Gamma$ and $\Gamma$ has an eigenvalue with multiplicity exactly two, then $n$ must be even.
\end{thm}
\begin{proof}
	Let $\theta$ be an eigenvalue with index set $M_\theta = \{r_1, r_2\}$. If zero transfer occurs between vertices $u$ and $v$, then~\Cref{lem:zero-criterion} gives
	$$(E_\theta)_{u,v} = \frac{1}{n} (\omega_n^{r_1 (u-v)}+\omega_n^{r_2 (u-v)}) = 0.$$
	Thus,
	\[
		\omega_n^{(r_1-r_2)(u-v)}=-1.
	\]
	Equivalently,
	\[
		\exp\left(\frac{2\pi \ii (r_1-r_2)(u-v)}{n}\right)=-1,
	\]
	which implies that
	\[
		\frac{2(r_1-r_2)(u-v)}{n}\in 2\mathbb Z+1.
	\]
	Hence $n$ divides an even integer but not an odd integer, and therefore $n$ must be even.
\end{proof}

For even $n$, the index $n/2$ plays a critical role in the spectrum of oriented circulant graphs. Notice that for $j = n/2$, by~\Cref{eq::22}, the eigenvalue formula gives:
$$\mu_{n/2} = -2\sum_{k \in \mathcal{C}} \sin\left(\frac{2\pi (n/2)k}{n}\right) = -2\sum_{k \in \mathcal{C}} \sin(\pi k) = 0.$$
This ensures that the index set $M_0$ corresponding to the eigenvalue $0$ always contains both $0$ and $n/2$. This leads to the following structural restriction on the distance between vertices exhibiting zero transfer.

\begin{cor}\label{cor::1}
	Let $\Gamma=G(\mathbb{Z}_n,\mathcal{C})$ be an oriented circulant graph with an even number of vertices. If there is zero transfer between $u$ and $v$ in $\Gamma$, and the eigenvalue $0$ has multiplicity exactly two, then $u-v$ is odd.
\end{cor}

\begin{proof}
	Since $\mu_0 = \mu_{n/2} = 0$, and the eigenvalue $0$ has multiplicity exactly two, its index set is $M_0 = \{0, n/2\}$. The spectral idempotent condition for zero transfer requires $(E_0)_{u,v} = 0$. Substituting the indices gives
	$$(E_0)_{u,v} = \frac{1}{n}\left(\omega_n^{0 \cdot (u-v)} + \omega_n^{(n/2) \cdot (u-v)}\right) = \frac{1}{n}\left(1 + (-1)^{u-v}\right) = 0.$$
	This equation holds if and only if $u-v$ is odd.
\end{proof}

Oriented circulant graphs $\Gamma = G(\mathbb{Z}_n,\mathcal{C})$ are vertex-transitive. Hence, if there is zero transfer between $u$ and $v$ for some vertices $u,v\in\mathbb{Z}_n$, then zero transfer also occurs between $u+k \pmod n$ and $v+k \pmod n$ for any integer $k$. Therefore, it suffices to consider zero transfer between $v$ and $0$. Let
\[
	S=\{\,v\in\mathbb{Z}_n : U(t)_{v,0}=0\text{ for all }t\in\mathbb R\,\}.
\]
The following lemma follows.

\begin{lem}
	Let $\Gamma=G(\mathbb{Z}_n,\mathcal{C})$ be an oriented circulant graph. If there is zero transfer between $v$ and $0$ in $\Gamma$, then there is also zero transfer between $n-v$ and $0$. Equivalently, $S=\overline{S}$, where $\overline{S}=\{n-s:\, s\in S\}$.
\end{lem}
\begin{proof}
	For any $v\in S$, zero transfer between $v$ and $0$ implies zero transfer between $v-v\pmod n$ and $0-v\pmod n$, that is, between $0$ and $n-v$. Hence $n-v\in S$.
\end{proof}

Furthermore, we can establish a necessary and sufficient condition for zero transfer at the antipodal vertex, that is, at distance $n/2$.

\begin{thm}\label{thm::antipodal}
	Let $\Gamma = G(\mathbb{Z}_n, \mathcal{C})$ be an oriented circulant graph with even order $n$.
	There is zero transfer between the antipodal vertex $n/2$ and vertex $0$ if and only if for every distinct eigenvalue $\theta_i$, the corresponding index set
	\[
		M_{\theta_i} = \{j \in \mathbb{Z}_n \mid \mu_j = \theta_i\}
	\]
	contains the same number of even and odd integers.
\end{thm}

\begin{proof}
	By~\Cref{lem:zero-criterion}, zero transfer occurs between $n/2$ and $0$ if and only if $(E_{\theta_i})_{n/2,0} = 0$ for every eigenvalue $\theta_i$. We expand the $(n/2,0)$-entry of the spectral idempotent:
	$$(E_{\theta_i})_{n/2,0} = \frac{1}{n} \sum_{j \in M_{\theta_i}} \omega_n^{j (n/2)} = \frac{1}{n} \sum_{j \in M_{\theta_i}} \exp(\ii \pi j) = \frac{1}{n} \sum_{j \in M_{\theta_i}} (-1)^j.$$
	Observe that $(-1)^j = 1$ if $j$ is even, and $(-1)^j = -1$ if $j$ is odd. Therefore, the sum $\sum_{j \in M_{\theta_i}} (-1)^j = 0$ if and only if the numbers of even and odd indices in $M_{\theta_i}$ are equal.
\end{proof}

\begin{cor}
	Let $\Gamma = G(\mathbb{Z}_n,\mathcal{C})$ be an oriented circulant graph with even order $n$.
	If zero transfer occurs between the antipodal vertex $n/2$ and $0$, then every eigenvalue has even multiplicity.
\end{cor}

\begin{proof}
	By Theorem~\ref{thm::antipodal}, zero transfer at the antipodal vertex requires that each index set $M_{\theta_i}$ contains the same number of even and odd integers.

	Therefore, we have
	\[
		|M_{\theta_i}|
		=
		|M_{\theta_i}\cap(2\mathbb Z)|
		+
		|M_{\theta_i}\cap(2\mathbb Z+1)|=2|M_{\theta_i}\cap(2\mathbb Z)|=2|M_{\theta_i}\cap(2\mathbb Z+1)|.
	\]

	Since $|M_{\theta_i}|$ equals the multiplicity of the eigenvalue $\theta_i$, every eigenvalue must have even multiplicity. In particular, no eigenvalue can be simple.
\end{proof}

\begin{cor}
	Let $\Gamma = G(\mathbb{Z}_n,\mathcal{C})$ be an oriented circulant graph of order $n$ with $n \equiv 0 \pmod 4$. If the eigenvalue $0$ has multiplicity two, then zero transfer cannot occur between the antipodal vertex $n/2$ and $0$.
\end{cor}

\begin{proof}
	Assume, to the contrary, that zero transfer occurs between $n/2$ and $0$.
	Since $\mu_0=\mu_{n/2}=0$ and the eigenvalue $0$ has multiplicity two, the corresponding index set is
	\[
		M_0=\{0,n/2\}.
	\]

	Because $n \equiv 0 \pmod 4$, the integer $n/2$ is even. Hence both elements of $M_0$ are even, and therefore
	\[
		|M_0 \cap (2\mathbb{Z})|=2,
		\qquad
		|M_0 \cap (2\mathbb{Z}+1)|=0 .
	\]

	However, by~\Cref{thm::antipodal}, zero transfer at the antipodal vertex requires that each index set $M_{\theta_i}$ contain the same number of even and odd integers. The set $M_0$ does not satisfy this condition. This contradiction shows that zero transfer between $n/2$ and $0$ is impossible.
\end{proof}

\begin{thm}\label{thm:4.4}
	Let $\Gamma = G(\mathbb{Z}_n,\mathcal{C})$ be an oriented circulant graph with $n \equiv 0 \pmod 4$. There is zero transfer between $n/4$ and $0$ if and only if, for every eigenvalue $\theta_i$ with index set $M_{\theta_i}$,
	\[|M_{\theta_i} \cap (4\mathbb{Z})|=|M_{\theta_i} \cap (4\mathbb{Z}+2)|
	\]
	and
	\[
		|M_{\theta_i} \cap (4\mathbb{Z}+1)|=|M_{\theta_i} \cap (4\mathbb{Z}+3)|.
	\]
\end{thm}

\begin{proof}
	By~\Cref{lem:zero-criterion}, zero transfer between $n/4$ and $0$ occurs if and only if
	\[
		(E_{\theta_i})_{n/4,0}=0
	\]
	for every eigenvalue $\theta_i$.

	We expand the $(n/4,0)$-entry of the spectral idempotent:
	\[
		(E_{\theta_i})_{n/4,0}
		=
		\frac{1}{n}
		\sum_{j\in M_{\theta_i}}\omega_n^{j(n/4)}
		=
		\frac{1}{n}
		\sum_{j\in M_{\theta_i}}
		\exp\!\left(\frac{\ii \pi j}{2}\right)
		=
		\frac{1}{n}\sum_{j\in M_{\theta_i}} \ii^j .
	\]

	Since
	\[
		\ii^j=
		\begin{cases}
			1,    & j\equiv0\pmod4, \\
			\ii,  & j\equiv1\pmod4, \\
			-1,   & j\equiv2\pmod4, \\
			-\ii, & j\equiv3\pmod4,
		\end{cases}
	\]
	the above sum can be written as
	\[
		\sum_{j\in M_{\theta_i}} \ii^j
		=
		\bigl(|M_{\theta_i}\cap(4\mathbb Z)|-|M_{\theta_i}\cap(4\mathbb Z+2)|\bigr)
		+
		\ii\bigl(|M_{\theta_i}\cap(4\mathbb Z+1)|-|M_{\theta_i}\cap(4\mathbb Z+3)|\bigr).
	\]

	Therefore $(E_{\theta_i})_{n/4,0}=0$ holds if and only if both the real and imaginary parts are zero, which implies
	\[
		|M_{\theta_i}\cap(4\mathbb Z)|
		=
		|M_{\theta_i}\cap(4\mathbb Z+2)|
	\]
	and
	\[
		|M_{\theta_i}\cap(4\mathbb Z+1)|
		=
		|M_{\theta_i}\cap(4\mathbb Z+3)|.
	\]
	This completes the proof.
\end{proof}
\begin{thm}\label{thm:4.5}
	Let $\Gamma = G(\mathbb{Z}_n, \mathcal{C})$ be an oriented circulant graph with $n \equiv 0 \pmod 4$, where every element of $\mathcal{C}$ is odd. Then:
	\begin{enumerate}[label = \bf(\roman*)]
		\item If $n \equiv 4 \pmod 8$, the real part of $(E_{\theta_i})_{n/4,0}$ is zero for every eigenvalue $\theta_i$. Zero transfer occurs between $n/4$ and $0$ if and only if $|M_{\theta_i} \cap (4\mathbb{Z}+1)| = |M_{\theta_i} \cap (4\mathbb{Z}+3)|$ for every $\theta_i$.
		\item If $n \equiv 0 \pmod 8$, the imaginary part of $(E_{\theta_i})_{n/4,0}$ is zero for every eigenvalue $\theta_i$. Zero transfer occurs between $n/4$ and $0$ if and only if $|M_{\theta_i} \cap 4\mathbb{Z}| = |M_{\theta_i} \cap (4\mathbb{Z}+2)|$ for every $\theta_i$.
	\end{enumerate}
\end{thm}

\begin{proof}
	By~\Cref{eq::22}, the eigenvalues of $\Gamma$ are given by
	\[
		\mu_j = -2 \sum_{k \in \mathcal{C}} \sin\left(\frac{2\pi j k}{n}\right).
	\]
	Since every $k \in \mathcal{C}$ is odd, we have
	\[
		\mu_{n/2-j}
		= -2 \sum_{k \in \mathcal{C}} \sin\left(\pi k - \frac{2\pi j k}{n}\right)
		= -2 \sum_{k \in \mathcal{C}} \sin\left(\frac{2\pi j k}{n}\right)
		= \mu_j.
	\]
	Hence, the map $f(j) = n/2 - j \pmod n$ preserves each index set $M_{\theta_i}$ and defines a bijection on it.

	We distinguish two cases according to $n \bmod 8$.

	\textbf{Case 1:} $n \equiv 4 \pmod{8}$. Then $n/2 \equiv 2 \pmod{4}$. For $j \equiv 0 \pmod{4}$,
	\[
		f(j) \equiv 2 \pmod{4}.
	\]
	Thus, $f$ induces a bijection between $M_{\theta_i} \cap 4\mathbb{Z}$ and $M_{\theta_i} \cap (4\mathbb{Z}+2)$, and hence
	\[
		|M_{\theta_i} \cap 4\mathbb{Z}| = |M_{\theta_i} \cap (4\mathbb{Z}+2)|.
	\]
	By~\Cref{thm:4.4}, the real part of the idempotent sum is zero, so only the imaginary part condition remains.

	\textbf{Case 2:} $n \equiv 0 \pmod{8}$. Then $n/2 \equiv 0 \pmod{4}$. For $j \equiv 1 \pmod{4}$,
	\[
		f(j) \equiv 3 \pmod{4}.
	\]
	Thus, $f$ induces a bijection between $M_{\theta_i} \cap (4\mathbb{Z}+1)$ and $M_{\theta_i} \cap (4\mathbb{Z}+3)$, and hence
	\[
		|M_{\theta_i} \cap (4\mathbb{Z}+1)| = |M_{\theta_i} \cap (4\mathbb{Z}+3)|.
	\]
	By~\Cref{thm:4.4}, the imaginary part of the idempotent sum is zero, so only the real part condition remains.
\end{proof}

\begin{cor}
	Let $\Gamma = G(\mathbb{Z}_n, \mathcal{C})$ be an oriented circulant graph with $n \equiv 0 \pmod 4$ and $\mathcal{C}$ consisting only of odd integers. If $\mathcal{C}$ can be partitioned into pairs $(k_1, k_2)$ such that:
	\begin{enumerate}[label = \bf(\roman*)]
		\item $k_2 - k_1 \equiv n/2 \pmod n$ when $n \equiv 4 \pmod 8$, or
		\item $k_1 + k_2 \equiv n/2 \pmod n$ when $n \equiv 0 \pmod 8$,
	\end{enumerate}
	then zero transfer occurs between $n/4$ and $0$.
\end{cor}

\begin{proof}
	Under the pairing conditions, we verify that the index sets satisfy the condition in~\Cref{thm:4.5}. For a pair $(k_1,k_2)$, write
	\[
		\mu_j^{(k_1,k_2)} = -2\Big(\sin\Big(\frac{2\pi j k_1}{n}\Big) + \sin\Big(\frac{2\pi j k_2}{n}\Big)\Big).
	\]

	\textbf{Case 1:} $n \equiv 4 \pmod{8}$.
	Here $k_2 = k_1 + n/2$, and
	\[
		\mu_j^{(k_1,k_2)}
		= -2\left(\sin\left(\frac{2\pi j k_1}{n}\right)
		+ \sin\left(\frac{2\pi j k_1}{n} + r\pi\right)\right)
		= -2 \sin\left(\frac{2\pi j k_1}{n}\right)\big(1 + (-1)^j\big).
	\]
	Hence $\mu_j^{(k_1,k_2)} = 0$ for all odd $j$, and therefore $\mu_j=0$ for all odd $j$. Thus all odd indices lie in $M_0$.

	Since $n/2 \equiv 2 \pmod{4}$, the map $j \mapsto j + n/2$ sends
	\[
		4\mathbb{Z}+1 \longleftrightarrow 4\mathbb{Z}+3,
	\]
	and preserves $M_0$. It follows that
	\[
		|M_0 \cap (4\mathbb{Z}+1)| = |M_0 \cap (4\mathbb{Z}+3)|.
	\]
	For $\theta \neq 0$, the set $M_\theta$ contains only even indices, which satisfy the required the condition by~\Cref{thm:4.5}.

	\textbf{Case 2:} $n \equiv 0 \pmod{8}$.
	Here $k_2 = n/2 - k_1$, and
	\[
		\mu_j^{(k_1,k_2)}
		= -2\left(\sin\left(\frac{2\pi j k_1}{n}\right)
		+ \sin\left(j\pi - \frac{2\pi j k_1}{n}\right)\right).
	\]
	If $j$ is even, then $j\pi \equiv 0 \pmod{2\pi}$, so $\sin(j\pi - x) = -\sin x$, and hence $\mu_j^{(k_1,k_2)}=0$. Therefore $\mu_j=0$ for all even $j$, and all even indices lie in $M_0$.

	Since $n \equiv 0 \pmod{8}$, the set of even indices has size $n/2$, and is evenly split between $4\mathbb{Z}$ and $4\mathbb{Z}+2$. Hence
	\[
		|M_0 \cap 4\mathbb{Z}| = |M_0 \cap (4\mathbb{Z}+2)| = \frac{n}{4}.
	\]
	For $\theta \neq 0$, the set $M_\theta$ contains only odd indices, which satisfy the required the condition by~\Cref{thm:4.5}.

	In both cases, the conditions of~\Cref{thm:4.5} are satisfied.
\end{proof}

\begin{con}
	Let $\Gamma = G(\mathbb{Z}_n, \mathcal{C})$ be an oriented circulant graph with $n \equiv 2 \pmod{4}$. If zero transfer occurs between vertex $v$ and $0$, then $v$ must be odd.
\end{con}

\noindent \textbf{Remark.} The necessity of this parity restriction is strongly suggested by computational evidence, although a general algebraic proof remains open. Let
\[
	M_0=\{j\in \mathbb Z_n:\mu_j=0\}
\]
denote the index set corresponding to the eigenvalue $0$. Since $\mu_j=-\mu_{n-j}$, we have $0,n/2\in M_0$. If $v$ is even, say $v=2m$, then the $(v,0)$-entry of the corresponding spectral idempotent satisfies
\[
	(E_0)_{v,0}
	=\frac1n\sum_{j\in M_0}\omega_n^{jv}
	=\frac1n\left(
	2+\sum_{j\in M_0\setminus\{0,n/2\}}
	\omega_n^{2mj}
	\right).
\]

Since $n\equiv2\pmod4$, the integer $n/2$ is odd, and hence each term $\omega_n^{2mj}$ lies in the odd cyclotomic field $\mathbb Q(\omega_{n/2})$. Therefore, zero transfer would require
\[
	\sum_{j\in M_0\setminus\{0,n/2\}}
	\omega_n^{2mj}=-2.
\]
While a single conjugate pair $\{j,n-j\}$ cannot yield such a contribution, it appears highly nontrivial to exclude the possibility that a larger union of symmetric pairs in $M_0$ could sum to $-2$.

Nevertheless, our computational search for all even integers $n\le 20$ (see Appendix) shows that whenever $n\equiv2\pmod4$, the zero transfer set $S$ consists entirely of odd integers. This leads us to conjecture that such an exact cancellation cannot occur in the relevant odd cyclotomic fields arising from connected oriented circulant graphs.

\begin{ex}
	Let $\Gamma$ be the oriented circulant graph $G(\mathbb{Z}_{12},\{2,3,8\})$ shown in Fig.~\ref{FIG}. Then zero transfer occurs between every vertex of $S=\{1,5,7,11\}$  and vertex $0$.
\end{ex}

\begin{figure}[htbp!]
	\centering
	\begin{tikzpicture}
		\foreach \x in {0,1,...,11}
			{
				\ifnum\x=0\relax
					\definecolor{nodecolor}{RGB}{220,70,70} % soft red
				\else\ifnum\x=1\relax
						\definecolor{nodecolor}{RGB}{100,149,237} % soft blue
					\else\ifnum\x=5\relax
							\definecolor{nodecolor}{RGB}{100,149,237}
						\else\ifnum\x=7\relax
								\definecolor{nodecolor}{RGB}{100,149,237}
							\else\ifnum\x=11\relax
									\definecolor{nodecolor}{RGB}{100,149,237}
								\else
									\definecolor{nodecolor}{RGB}{230,230,230} % light gray
								\fi\fi\fi\fi\fi

				\node[
					circle,
					fill=nodecolor,
					draw=black!70,
					line width=0.5pt,
					inner sep=0pt,
					minimum size=5mm
				] (a\x) at (-\x*30:2.5){\x};
			}
		\foreach \x in {0,1,...,11}
			{
				\pgfmathparse{int(mod(\x+2,12))}
				\draw[-stealth] (a\x) -- (a\pgfmathresult);
			}
		\foreach \x in {0,1,...,11}
			{
				\foreach[parse] \y in{\x+3,\x+8}
					{
						\pgfmathparse{int(mod(\y,12))}
						\draw[-stealth] (a\x) -- (a\pgfmathresult);
					}
			}
	\end{tikzpicture}
	\caption{Zero transfer in $G(\mathbb{Z}_{12},\{2,3,8\})$. \label{FIG}}
\end{figure}
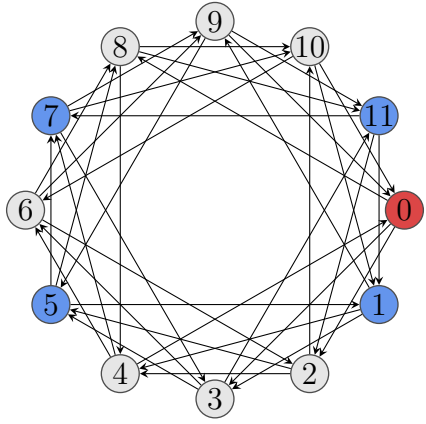

While the parity restrictions and antipodal conditions provide analytical criteria for specific cases, determining zero transfer for graphs with higher eigenvalue multiplicities involves resolving complex vanishing sums of roots of unity. To better understand the distribution and frequency of zero transfer in these cases, we conducted a computational search over connected oriented circulant graphs for even $n \le 20$. The counts are summarized in Table~\ref{tab:sample}.

\begin{table}[htbp]
	\centering
	\caption{Counts of zero transfer in oriented circulant graphs}
	\label{tab:sample}
	\begin{tabular}{ccccccccc}
		\toprule
		\textbf{Order} & \textbf{6} & \textbf{8} & \textbf{10} & \textbf{12} & \textbf{14} & \textbf{16} & \textbf{18} & \textbf{20} \\ \midrule
		\textbf{Count} & 2          & 2          & 24          & 10          & 74          & 14          & 310         & 146         \\ \bottomrule
	\end{tabular}
\end{table}

\noindent \textbf{Remark on Computational Results.} The data presented in Table~\ref{tab:sample} exhibit a clear parity phenomenon. The frequency of zero transfer instances is significantly higher when $n \equiv 2 \pmod 4$ (for example, $n=10,14,18$) than for more highly divisible even orders with $n \equiv 0 \pmod 4$ (for example, $n=8,16$).

\section{Conclusions}

In this paper, we investigated zero transfer on mixed graphs, with particular emphasis on oriented circulant graphs. Using Taylor expansion and spectral decomposition, we derived a general algebraic characterization of zero transfer (Theorem~\ref{thm:2}), showing that it is completely determined by the vanishing entries of the spectral idempotents. We also proved that connected undirected graphs cannot exhibit zero transfer by the Perron--Frobenius theorem, implying that directed, signed, or mixed structures are essential for this phenomenon in connected networks.

For oriented circulant graphs, we showed that zero transfer is strongly constrained by both the graph order and the eigenvalue multiplicities. In particular, connected oriented circulant graphs of prime order do not admit zero transfer (Theorem~\ref{thm:11}) because they necessarily possess a simple eigenvalue. Although this obstruction extends to small odd composite orders, we identified $n=21$ as the smallest odd order admitting zero transfer.

For even orders, the possible zero transfer vertices satisfy strong parity restrictions. If an eigenvalue has multiplicity two, then zero transfer can occur only between vertices at odd distance (Corollary~\ref{cor::1}). We further characterized zero transfer at the antipodal vertex (Theorem~\ref{thm::antipodal}) and at distance $n/4$ (Theorem~\ref{thm:4.4}) through parity-balanced eigenspace conditions. Finally, supported by exhaustive computations for $n \le 20$, we conjectured that when $n \equiv 2 \pmod 4$, zero transfer occurs exclusively at odd distances.

\vspace{0.5em}

\noindent\textbf{Future Work:} Zero Transfer and Number-Theoretic Structures

Computational evidence suggests a close connection between the zero transfer set $S$ and the arithmetic structure of the cyclic group $\mathbb{Z}_n$. In several examples, the zero transfer vertices coincide with the reduced residue system modulo $n$, namely
\[
	S=\{v\in \mathbb{Z}_n:\gcd(v,n)=1\}.
\]

This observation indicates a possible link between zero transfer and the cyclotomic structure of the eigenvalue equations. Indeed, if
\[
	\sum_{j\in M_r}\omega_n^j=0,
\]
then the primitive $n$-th root of unity $\omega_n$ is a root of the polynomial
\[
	P(x)=\sum_{j\in M_r}x^j.
\]
Since the cyclotomic polynomial $\Phi_n(x)$ is irreducible over $\mathbb{Q}$, one expects that, under suitable conditions, all primitive $n$-th roots $\omega_n^v$ with $\gcd(v,n)=1$ must also satisfy the same relation. This may explain the appearance of the reduced residue system in the zero transfer set.

A natural direction for future work is to determine the precise algebraic conditions on the generating set $\mathcal{C}$ under which this phenomenon occurs, potentially using techniques from cyclotomic fields, character theory, and representation theory of finite abelian groups.

\appendix % 开始附录

\renewcommand\thesection{\Alph{section}}
\setcounter{section}{0}
\section{Appendix}\label{Appendix}

\textit{Data description:}
Let $G(\mathbb{Z}_n, \mathcal{C})$ be an oriented circulant graph. The tables below list all cases of zero transfer in oriented circulant graphs for $n \le 20$. The data are categorized by the number of vertices $n$ and the set of vertices $S$ that exhibit zero transfer with vertex $0$.

\subsection*{$n=6$.}
\begin{itemize}[label=\textbullet,leftmargin=1.15em,itemsep=0.65em,topsep=0.12em,parsep=0pt,partopsep=0pt]
	\item \textbf{$S=\{3\}$}, $\lvert\mathcal{C}\rvert=2$.

	      $\mathcal{C}\in$\{\{1\}, \{5\}\}.

\end{itemize}
\subsection*{$n=8$.}
\begin{itemize}[label=\textbullet,leftmargin=1.15em,itemsep=0.65em,topsep=0.12em,parsep=0pt,partopsep=0pt]
	\item \textbf{$S=\{2, 6\}$}, $\lvert\mathcal{C}\rvert=2$.

	      $\mathcal{C}\in$\{\{1, 3\}, \{5, 7\} \}.

\end{itemize}
\subsection*{$n=10$.}
\begin{itemize}[label=\textbullet,leftmargin=1.15em,itemsep=0.65em,topsep=0.12em,parsep=0pt,partopsep=0pt]
	\item \textbf{$S=\{5\}$}, $\lvert\mathcal{C}\rvert=24$.

	      $\mathcal{C}\in$\{\{1\}, \{3\}, \{7\}, \{9\}, \{1, 2\}, \{1, 3\}, \{1, 7\}, \{1, 8\}, \{2, 9\}, \{3, 4\}, \{3, 6\}, \{3, 9\}, \{4, 7\}, \{6, 7\}, \{7, 9\}, \{8, 9\}, \{1, 2, 3, 6\}, \{1, 2, 4, 7\}, \{1, 3, 4, 8\}, \{1, 6, 7, 8\}, \{2, 3, 4, 9\}, \{2, 6, 7, 9\}, \{3, 6, 8, 9\}, \{4, 7, 8, 9\} \}.
\end{itemize}
\subsection*{$n=12$.}
\begin{itemize}[label=\textbullet,leftmargin=1.15em,itemsep=0.65em,topsep=0.12em,parsep=0pt,partopsep=0pt]
	\item \textbf{$S=\{1, 5, 7, 11\}$}, $\lvert\mathcal{C}\rvert=4$.

	      $\mathcal{C}\in$ \{ \{2, 3, 8\}, \{2, 8, 9\}, \{3, 4, 10\}, \{4, 9, 10\} \}.

	\item \textbf{$S=\{3, 6, 9\}$}, $\lvert\mathcal{C}\rvert=4$.

	      $\mathcal{C}\in$\{ \{1, 2, 4, 7\}, \{1, 7, 8, 10\}, \{2, 4, 5, 11\}, \{5, 8, 10, 11\} \}.

	\item \textbf{$S=\{3, 9\}$}, $\lvert\mathcal{C}\rvert=2$.

	      $\mathcal{C}\in$ \{ \{1, 7\}, \{5, 11\} \}.

\end{itemize}
\subsection*{$n=14$.}
\begin{itemize}[label=\textbullet,leftmargin=1.15em,itemsep=0.65em,topsep=0.12em,parsep=0pt,partopsep=0pt]
	\item \textbf{$S=\{7\}$}, $\lvert\mathcal{C}\rvert=74$.

	      $\mathcal{C}\in$\{ \{1\}, \{3\}, \{5\}, \{9\}, \{11\}, \{13\}, \{1, 3\}, \{1, 5\}, \{1, 9\}, \{1, 11\}, \{3, 5\}, \{3, 9\}, \{3, 13\}, \{5, 11\}, \{5, 13\}, \{9, 11\}, \{9, 13\}, \{11, 13\}, \{1, 2, 3\}, \{1, 2, 10\}, \{1, 3, 5\}, \{1, 3, 9\}, \{1, 3, 12\}, \{1, 4, 5\}, \{1, 4, 12\}, \{1, 5, 10\}, \{1, 5, 11\}, \{1, 9, 11\}, \{2, 3, 6\}, \{2, 6, 11\}, \{2, 10, 13\}, \{2, 11, 13\}, \{3, 5, 13\}, \{3, 6, 9\}, \{3, 8, 9\}, \{3, 8, 12\}, \{3, 9, 13\}, \{4, 5, 6\}, \{4, 6, 9\}, \{4, 9, 13\}, \{4, 12, 13\}, \{5, 6, 11\}, \{5, 8, 10\}, \{5, 8, 11\}, \{5, 11, 13\}, \{8, 9, 10\}, \{8, 11, 12\}, \{9, 10, 13\}, \{9, 11, 13\}, \{11, 12, 13\}, \{1, 2, 5, 8\}, \{1, 2, 6, 9\}, \{1, 3, 4, 8\}, \{1, 3, 6, 10\}, \{1, 4, 6, 11\}, \{1, 5, 6, 12\}, \{1, 8, 9, 12\}, \{1, 8, 10, 11\}, \{2, 3, 4, 9\}, \{2, 3, 5, 10\}, \{2, 4, 5, 11\}, \{2, 5, 6, 13\}, \{2, 8, 9, 13\}, \{2, 9, 10, 11\}, \{3, 4, 5, 12\}, \{3, 4, 6, 13\}, \{3, 8, 10, 13\}, \{3, 9, 10, 12\}, \{4, 8, 11, 13\}, \{4, 9, 11, 12\}, \{5, 8, 12, 13\}, \{5, 10, 11, 12\}, \{6, 9, 12, 13\}, \{6, 10, 11, 13\} \}.

\end{itemize}
\subsection*{$n=16$.}
\begin{itemize}[label=\textbullet,leftmargin=1.15em,itemsep=0.65em,topsep=0.12em,parsep=0pt,partopsep=0pt]
	\item \textbf{$S=\{2, 6, 10, 14\}$}, $\lvert\mathcal{C}\rvert=6$.

	      $\mathcal{C}\in$\{ \{1, 3, 9, 11\}, \{5, 7, 13, 15\}, \{1, 3, 4, 9, 11\}, \{1, 3, 9, 11, 12\}, \{4, 5, 7, 13, 15\}, \{5, 7, 12, 13, 15\} \}.

	\item \textbf{$S=\{4, 12\}$}, $\lvert\mathcal{C}\rvert=8$.

	      $\mathcal{C}\in$\{ \{1, 7\}, \{3, 5\}, \{9, 15\}, \{11, 13\}, \{1, 3, 5, 7\}, \{1, 7, 11, 13\}, \{3, 5, 9, 15\}, \{9, 11, 13, 15\} \}.

\end{itemize}
\subsection*{$n=18$.}
\begin{itemize}[label=\textbullet,leftmargin=1.15em,itemsep=0.65em,topsep=0.12em,parsep=0pt,partopsep=0pt]
	\item \textbf{$S=\{1, 5, 7, 9, 11, 13, 17\}$}, $\lvert\mathcal{C}\rvert=4$.

	      $\mathcal{C}\in$\{ \{2, 3, 8, 14\}, \{2, 8, 14, 15\}, \{3, 4, 10, 16\}, \{4, 10, 15, 16\} \}.

	\item \textbf{$S=\{1, 5, 7, 11, 13, 17\}$}, $\lvert\mathcal{C}\rvert=8$.

	      $\mathcal{C}\in$\{ \{2, 3, 6, 8, 14\}, \{2, 3, 8, 12, 14\}, \{2, 6, 8, 14, 15\}, \{2, 8, 12, 14, 15\}, \{3, 4, 6, 10, 16\}, \{3, 4, 10, 12, 16\}, \{4, 6, 10, 15, 16\}, \{4, 10, 12, 15, 16\} \}.

	\item \textbf{$S=\{3, 9, 15\}$}, $\lvert\mathcal{C}\rvert=114$.

	      $\mathcal{C}\in$\{ \{1, 2, 5\}, \{1, 4, 11\}, \{1, 5, 16\}, \{1, 7, 13\}, \{1, 11, 14\}, \{2, 13, 17\}, \{4, 7, 17\}, \{5, 7, 8\}, \{5, 7, 10\}, \{5, 11, 17\}, \{7, 14, 17\}, \{8, 11, 13\}, \{10, 11, 13\}, \{13, 16, 17\}, \{1, 2, 5, 6\}, \{1, 2, 5, 12\}, \{1, 2, 8, 11\}, \{1, 4, 5, 10\}, \{1, 4, 6, 11\}, \{1, 4, 11, 12\}, \{1, 5, 6, 16\}, \{1, 5, 8, 14\}, \{1, 5, 12, 16\}, \{1, 6, 7, 13\}, \{1, 6, 11, 14\}, \{1, 7, 12, 13\}, \{1, 10, 11, 16\}, \{1, 11, 12, 14\}, \{2, 5, 7, 14\}, \{2, 6, 13, 17\}, \{2, 7, 8, 17\}, \{2, 11, 13, 14\}, \{2, 12, 13, 17\}, \{4, 5, 7, 16\}, \{4, 6, 7, 17\}, \{4, 7, 12, 17\}, \{4, 10, 13, 17\}, \{4, 11, 13, 16\}, \{5, 6, 7, 8\}, \{5, 6, 7, 10\}, \{5, 6, 11, 17\}, \{5, 7, 8, 12\}, \{5, 7, 10, 12\}, \{5, 11, 12, 17\}, \{6, 7, 14, 17\}, \{6, 8, 11, 13\}, \{6, 10, 11, 13\}, \{6, 13, 16, 17\}, \{7, 10, 16, 17\}, \{7, 12, 14, 17\}, \{8, 11, 12, 13\}, \{8, 13, 14, 17\}, \{10, 11, 12, 13\}, \{12, 13, 16, 17\}, \{1, 2, 4, 7, 13\}, \{1, 2, 6, 8, 11\}, \{1, 2, 7, 10, 13\}, \{1, 2, 8, 11, 12\}, \{1, 4, 5, 6, 10\}, \{1, 4, 5, 10, 12\}, \{1, 4, 7, 8, 13\}, \{1, 5, 6, 8, 14\}, \{1, 5, 8, 12, 14\}, \{1, 6, 10, 11, 16\}, \{1, 7, 8, 13, 16\}, \{1, 7, 10, 13, 14\}, \{1, 7, 13, 14, 16\}, \{1, 10, 11, 12, 16\}, \{2, 4, 5, 11, 17\}, \{2, 5, 6, 7, 14\}, \{2, 5, 7, 12, 14\}, \{2, 5, 10, 11, 17\}, \{2, 6, 7, 8, 17\}, \{2, 6, 11, 13, 14\}, \{2, 7, 8, 12, 17\}, \{2, 11, 12, 13, 14\}, \{4, 5, 6, 7, 16\}, \{4, 5, 7, 12, 16\}, \{4, 5, 8, 11, 17\}, \{4, 6, 10, 13, 17\}, \{4, 6, 11, 13, 16\}, \{4, 10, 12, 13, 17\}, \{4, 11, 12, 13, 16\}, \{5, 8, 11, 16, 17\}, \{5, 10, 11, 14, 17\}, \{5, 11, 14, 16, 17\}, \{6, 7, 10, 16, 17\}, \{6, 8, 13, 14, 17\}, \{7, 10, 12, 16, 17\}, \{8, 12, 13, 14, 17\}, \{1, 2, 4, 6, 7, 13\}, \{1, 2, 4, 7, 12, 13\}, \{1, 2, 6, 7, 10, 13\}, \{1, 2, 7, 10, 12, 13\}, \{1, 4, 6, 7, 8, 13\}, \{1, 4, 7, 8, 12, 13\}, \{1, 6, 7, 8, 13, 16\}, \{1, 6, 7, 10, 13, 14\}, \{1, 6, 7, 13, 14, 16\}, \{1, 7, 8, 12, 13, 16\}, \{1, 7, 10, 12, 13, 14\}, \{1, 7, 12, 13, 14, 16\}, \{2, 4, 5, 6, 11, 17\}, \{2, 4, 5, 11, 12, 17\}, \{2, 5, 6, 10, 11, 17\}, \{2, 5, 10, 11, 12, 17\}, \{4, 5, 6, 8, 11, 17\}, \{4, 5, 8, 11, 12, 17\}, \{5, 6, 8, 11, 16, 17\}, \{5, 6, 10, 11, 14, 17\}, \{5, 6, 11, 14, 16, 17\}, \{5, 8, 11, 12, 16, 17\}, \{5, 10, 11, 12, 14, 17\}, \{5, 11, 12, 14, 16, 17\} \}.

	\item \textbf{$S=\{9\}$}, $\lvert\mathcal{C}\rvert=184$.

	      $\mathcal{C}\in$\{ \{1\}, \{5\}, \{7\}, \{11\}, \{13\}, \{17\}, \{1, 3\}, \{1, 5\}, \{1, 7\}, \{1, 11\}, \{1, 13\}, \{1, 15\}, \{3, 5\}, \{3, 7\}, \{3, 11\}, \{3, 13\}, \{3, 17\}, \{5, 7\}, \{5, 11\}, \{5, 15\}, \{5, 17\}, \{7, 13\}, \{7, 15\}, \{7, 17\}, \{11, 13\}, \{11, 15\}, \{11, 17\}, \{13, 15\}, \{13, 17\}, \{15, 17\}, \{1, 2, 4\}, \{1, 3, 5\}, \{1, 3, 7\}, \{1, 3, 11\}, \{1, 3, 13\}, \{1, 5, 7\}, \{1, 5, 11\}, \{1, 5, 15\}, \{1, 7, 15\}, \{1, 11, 13\}, \{1, 11, 15\}, \{1, 13, 15\}, \{1, 14, 16\}, \{2, 4, 17\}, \{2, 5, 10\}, \{2, 10, 13\}, \{3, 5, 7\}, \{3, 5, 11\}, \{3, 5, 17\}, \{3, 7, 13\}, \{3, 7, 17\}, \{3, 11, 13\}, \{3, 11, 17\}, \{3, 13, 17\}, \{4, 7, 8\}, \{4, 8, 11\}, \{5, 7, 15\}, \{5, 7, 17\}, \{5, 8, 16\}, \{5, 11, 15\}, \{5, 15, 17\}, \{7, 10, 14\}, \{7, 13, 15\}, \{7, 13, 17\}, \{7, 15, 17\}, \{8, 13, 16\}, \{10, 11, 14\}, \{11, 13, 15\}, \{11, 13, 17\}, \{11, 15, 17\}, \{13, 15, 17\}, \{14, 16, 17\}, \{1, 2, 3, 4\}, \{1, 2, 4, 15\}, \{1, 2, 7, 10\}, \{1, 3, 5, 7\}, \{1, 3, 5, 11\}, \{1, 3, 7, 13\}, \{1, 3, 11, 13\}, \{1, 3, 14, 16\}, \{1, 4, 8, 13\}, \{1, 5, 7, 15\}, \{1, 5, 11, 15\}, \{1, 7, 8, 16\}, \{1, 7, 13, 15\}, \{1, 10, 13, 14\}, \{1, 11, 13, 15\}, \{1, 14, 15, 16\}, \{2, 3, 4, 17\}, \{2, 3, 5, 10\}, \{2, 3, 10, 13\}, \{2, 4, 5, 11\}, \{2, 4, 7, 13\}, \{2, 4, 15, 17\}, \{2, 5, 10, 15\}, \{2, 10, 11, 17\}, \{2, 10, 13, 15\}, \{3, 4, 7, 8\}, \{3, 4, 8, 11\}, \{3, 5, 7, 17\}, \{3, 5, 8, 16\}, \{3, 5, 11, 17\}, \{3, 7, 10, 14\}, \{3, 7, 13, 17\}, \{3, 8, 13, 16\}, \{3, 10, 11, 14\}, \{3, 11, 13, 17\}, \{3, 14, 16, 17\}, \{4, 5, 8, 17\}, \{4, 7, 8, 15\}, \{4, 8, 11, 15\}, \{5, 7, 15, 17\}, \{5, 8, 15, 16\}, \{5, 10, 14, 17\}, \{5, 11, 14, 16\}, \{5, 11, 15, 17\}, \{7, 10, 14, 15\}, \{7, 13, 14, 16\}, \{7, 13, 15, 17\}, \{8, 11, 16, 17\}, \{8, 13, 15, 16\}, \{10, 11, 14, 15\}, \{11, 13, 15, 17\}, \{14, 15, 16, 17\}, \{1, 2, 3, 7, 10\}, \{1, 2, 5, 8, 14\}, \{1, 2, 7, 10, 15\}, \{1, 2, 8, 11, 14\}, \{1, 3, 4, 8, 13\}, \{1, 3, 7, 8, 16\}, \{1, 3, 10, 13, 14\}, \{1, 4, 5, 10, 16\}, \{1, 4, 8, 13, 15\}, \{1, 4, 10, 11, 16\}, \{1, 7, 8, 15, 16\}, \{1, 10, 13, 14, 15\}, \{2, 3, 4, 5, 11\}, \{2, 3, 4, 7, 13\}, \{2, 3, 10, 11, 17\}, \{2, 4, 5, 11, 15\}, \{2, 4, 7, 13, 15\}, \{2, 5, 7, 8, 14\}, \{2, 7, 8, 14, 17\}, \{2, 8, 11, 13, 14\}, \{2, 8, 13, 14, 17\}, \{2, 10, 11, 15, 17\}, \{3, 4, 5, 8, 17\}, \{3, 5, 10, 14, 17\}, \{3, 5, 11, 14, 16\}, \{3, 7, 13, 14, 16\}, \{3, 8, 11, 16, 17\}, \{4, 5, 7, 10, 16\}, \{4, 5, 8, 15, 17\}, \{4, 7, 10, 16, 17\}, \{4, 10, 11, 13, 16\}, \{4, 10, 13, 16, 17\}, \{5, 10, 14, 15, 17\}, \{5, 11, 14, 15, 16\}, \{7, 13, 14, 15, 16\}, \{8, 11, 15, 16, 17\}, \{1, 2, 3, 5, 8, 14\}, \{1, 2, 3, 8, 11, 14\}, \{1, 2, 5, 8, 14, 15\}, \{1, 2, 8, 11, 14, 15\}, \{1, 3, 4, 5, 10, 16\}, \{1, 3, 4, 10, 11, 16\}, \{1, 4, 5, 10, 15, 16\}, \{1, 4, 10, 11, 15, 16\}, \{2, 3, 5, 7, 8, 14\}, \{2, 3, 7, 8, 14, 17\}, \{2, 3, 8, 11, 13, 14\}, \{2, 3, 8, 13, 14, 17\}, \{2, 5, 7, 8, 14, 15\}, \{2, 7, 8, 14, 15, 17\}, \{2, 8, 11, 13, 14, 15\}, \{2, 8, 13, 14, 15, 17\}, \{3, 4, 5, 7, 10, 16\}, \{3, 4, 7, 10, 16, 17\}, \{3, 4, 10, 11, 13, 16\}, \{3, 4, 10, 13, 16, 17\}, \{4, 5, 7, 10, 15, 16\}, \{4, 7, 10, 15, 16, 17\}, \{4, 10, 11, 13, 15, 16\}, \{4, 10, 13, 15, 16, 17\} \}.

\end{itemize}
\subsection*{$n=20$.}
\begin{itemize}[label=\textbullet,leftmargin=1.15em,itemsep=0.65em,topsep=0.12em,parsep=0pt,partopsep=0pt]
	\item \textbf{$S=\{1, 3, 7, 9, 11, 13, 17, 19\}$}, $\lvert\mathcal{C}\rvert=16$.

	      $\mathcal{C}\in$\{ \{2, 5, 12\}, \{2, 12, 15\}, \{4, 5, 14\}, \{4, 14, 15\}, \{5, 6, 16\}, \{5, 8, 18\}, \{6, 15, 16\}, \{8, 15, 18\}, \{2, 4, 5, 12, 14\}, \{2, 4, 12, 14, 15\}, \{2, 5, 6, 12, 16\}, \{2, 6, 12, 15, 16\}, \{4, 5, 8, 14, 18\}, \{4, 8, 14, 15, 18\}, \{5, 6, 8, 16, 18\}, \{6, 8, 15, 16, 18\} \}.

	\item \textbf{$S=\{5, 10, 15\}$}, $\lvert\mathcal{C}\rvert=32$.

	      $\mathcal{C}\in$\{\{1, 2, 8, 11\}, \{1, 4, 6, 11\}, \{1, 11, 12, 18\}, \{1, 11, 14, 16\}, \{2, 3, 8, 13\}, \{2, 7, 8, 17\}, \{2, 8, 9, 19\}, \{3, 4, 6, 13\}, \{3, 12, 13, 18\}, \{3, 13, 14, 16\}, \{4, 6, 7, 17\}, \{4, 6, 9, 19\}, \{7, 12, 17, 18\}, \{7, 14, 16, 17\}, \{9, 12, 18, 19\}, \{9, 14, 16, 19\}, \{1, 2, 3, 4, 6, 8, 11, 13\}, \{1, 2, 3, 8, 11, 13, 14, 16\}, \{1, 2, 4, 6, 7, 8, 11, 17\}, \{1, 2, 7, 8, 11, 14, 16, 17\}, \{1, 3, 4, 6, 11, 12, 13, 18\}, \{1, 3, 11, 12, 13, 14, 16, 18\}, \{1, 4, 6, 7, 11, 12, 17, 18\}, \{1, 7, 11, 12, 14, 16, 17, 18\}, \{2, 3, 4, 6, 8, 9, 13, 19\}, \{2, 3, 8, 9, 13, 14, 16, 19\}, \{2, 4, 6, 7, 8, 9, 17, 19\}, \{2, 7, 8, 9, 14, 16, 17, 19\}, \{3, 4, 6, 9, 12, 13, 18, 19\}, \{3, 9, 12, 13, 14, 16, 18, 19\}, \{4, 6, 7, 9, 12, 17, 18, 19\}, \{7, 9, 12, 14, 16, 17, 18, 19\} \}.

	\item \textbf{$S=\{5, 15\}$}, $\lvert\mathcal{C}\rvert=98$.

	      $\mathcal{C}\in$\{ \{1, 3\}, \{1, 7\}, \{1, 11\}, \{3, 9\}, \{3, 13\}, \{7, 9\}, \{7, 17\}, \{9, 19\}, \{11, 13\}, \{11, 17\}, \{13, 19\}, \{17, 19\}, \{1, 2, 11, 12\}, \{1, 3, 7, 9\}, \{1, 3, 11, 13\}, \{1, 7, 11, 17\}, \{1, 8, 11, 18\}, \{2, 9, 12, 19\}, \{3, 4, 13, 14\}, \{3, 6, 13, 16\}, \{3, 9, 13, 19\}, \{4, 7, 14, 17\}, \{6, 7, 16, 17\}, \{7, 9, 17, 19\}, \{8, 9, 18, 19\}, \{11, 13, 17, 19\}, \{1, 2, 3, 7, 9, 12\}, \{1, 2, 3, 8, 11, 13\}, \{1, 2, 4, 6, 8, 11\}, \{1, 2, 4, 6, 11, 12\}, \{1, 2, 7, 8, 11, 17\}, \{1, 2, 8, 11, 14, 16\}, \{1, 2, 11, 12, 14, 16\}, \{1, 3, 4, 6, 11, 13\}, \{1, 3, 4, 7, 9, 14\}, \{1, 3, 6, 7, 9, 16\}, \{1, 3, 7, 8, 9, 18\}, \{1, 3, 11, 12, 13, 18\}, \{1, 3, 11, 13, 14, 16\}, \{1, 4, 6, 7, 11, 17\}, \{1, 4, 6, 8, 11, 18\}, \{1, 4, 6, 11, 12, 18\}, \{1, 7, 11, 12, 17, 18\}, \{1, 7, 11, 14, 16, 17\}, \{1, 8, 11, 14, 16, 18\}, \{1, 11, 12, 14, 16, 18\}, \{2, 3, 4, 6, 8, 13\}, \{2, 3, 4, 8, 13, 14\}, \{2, 3, 6, 8, 13, 16\}, \{2, 3, 8, 9, 13, 19\}, \{2, 3, 8, 13, 14, 16\}, \{2, 4, 6, 7, 8, 17\}, \{2, 4, 6, 8, 9, 19\}, \{2, 4, 6, 9, 12, 19\}, \{2, 4, 7, 8, 14, 17\}, \{2, 6, 7, 8, 16, 17\}, \{2, 7, 8, 9, 17, 19\}, \{2, 7, 8, 14, 16, 17\}, \{2, 8, 9, 14, 16, 19\}, \{2, 9, 12, 14, 16, 19\}, \{2, 11, 12, 13, 17, 19\}, \{3, 4, 6, 9, 13, 19\}, \{3, 4, 6, 12, 13, 18\}, \{3, 4, 12, 13, 14, 18\}, \{3, 6, 12, 13, 16, 18\}, \{3, 9, 12, 13, 18, 19\}, \{3, 9, 13, 14, 16, 19\}, \{3, 12, 13, 14, 16, 18\}, \{4, 6, 7, 9, 17, 19\}, \{4, 6, 7, 12, 17, 18\}, \{4, 6, 8, 9, 18, 19\}, \{4, 6, 9, 12, 18, 19\}, \{4, 7, 12, 14, 17, 18\}, \{4, 11, 13, 14, 17, 19\}, \{6, 7, 12, 16, 17, 18\}, \{6, 11, 13, 16, 17, 19\}, \{7, 9, 12, 17, 18, 19\}, \{7, 9, 14, 16, 17, 19\}, \{7, 12, 14, 16, 17, 18\}, \{8, 9, 14, 16, 18, 19\}, \{8, 11, 13, 17, 18, 19\}, \{9, 12, 14, 16, 18, 19\}, \{1, 2, 3, 4, 7, 9, 12, 14\}, \{1, 2, 3, 6, 7, 9, 12, 16\}, \{1, 2, 3, 6, 11, 12, 13, 16\}, \{1, 2, 4, 7, 11, 12, 14, 17\}, \{1, 3, 4, 7, 8, 9, 14, 18\}, \{1, 3, 4, 8, 11, 13, 14, 18\}, \{1, 3, 6, 7, 8, 9, 16, 18\}, \{1, 6, 7, 8, 11, 16, 17, 18\}, \{2, 3, 4, 9, 12, 13, 14, 19\}, \{2, 4, 11, 12, 13, 14, 17, 19\}, \{2, 6, 7, 9, 12, 16, 17, 19\}, \{2, 6, 11, 12, 13, 16, 17, 19\}, \{3, 6, 8, 9, 13, 16, 18, 19\}, \{4, 7, 8, 9, 14, 17, 18, 19\}, \{4, 8, 11, 13, 14, 17, 18, 19\}, \{6, 8, 11, 13, 16, 17, 18, 19\} \}.
\end{itemize}

\begin{thebibliography}{18}
	\expandafter\ifx\csname natexlab\endcsname\relax\def\natexlab#1{#1}\fi
	\providecommand{\url}[1]{\texttt{#1}}
	\providecommand{\href}[2]{#2}
	\providecommand{\path}[1]{#1}
	\providecommand{\DOIprefix}{doi:}
	\providecommand{\ArXivprefix}{arXiv:}
	\providecommand{\URLprefix}{URL: }
	\providecommand{\Pubmedprefix}{pmid:}
	\providecommand{\doi}[1]{\href{http://dx.doi.org/#1}{\path{#1}}}
	\providecommand{\Pubmed}[1]{\href{pmid:#1}{\path{#1}}}
	\providecommand{\bibinfo}[2]{#2}
	\ifx\xfnm\relax \def\xfnm[#1]{\unskip,\space#1}\fi
	%Type = Article
	\bibitem[{Ben~Taher and Rachidi(2002)}]{BR02}
	\bibinfo{author}{R.~Ben~Taher}, \bibinfo{author}{M.~Rachidi},
	\newblock \bibinfo{title}{Some explicit formulas for the polynomial
		decomposition of the matrix exponential and applications},
	\newblock \bibinfo{journal}{\emph{Linear Algebra Appl.}}   \textbf{350}
	(\bibinfo{year}{2002}) \bibinfo{pages}{171--184}.
	%Type = Article
	\bibitem[{Bose(2003)}]{Bo03}
	\bibinfo{author}{S.~Bose},
	\newblock \bibinfo{title}{Quantum communication through an unmodulated spin
		chain},
	\newblock \bibinfo{journal}{\emph{Phys. Rev. Lett.}}   \textbf{91}
	(\bibinfo{year}{2003}) \bibinfo{pages}{207901}.
	%Type = Article
	\bibitem[{Cameron et~al.(2014)Cameron, Fehrenbach, Granger, Hennigh, Shrestha,
			and Tamon}]{CFGHST14}
	\bibinfo{author}{S.~Cameron}, \bibinfo{author}{S.~Fehrenbach},
	\bibinfo{author}{L.~Granger}, \bibinfo{author}{O.~Hennigh},
	\bibinfo{author}{S.~Shrestha}, \bibinfo{author}{C.~Tamon},
	\newblock \bibinfo{title}{Universal state transfer on graphs},
	\newblock \bibinfo{journal}{\emph{Linear Algebra Appl.}}   \textbf{455}
	(\bibinfo{year}{2014}) \bibinfo{pages}{115--142}.
	%Type = Article
	\bibitem[{Chaves et~al.(2023)Chaves, Chagas, and Coutinho}]{CCC23}
	\bibinfo{author}{R.~Chaves}, \bibinfo{author}{B.~Chagas},
	\bibinfo{author}{G.~Coutinho},
	\newblock \bibinfo{title}{Why and how to add direction to a quantum walk},
	\newblock \bibinfo{journal}{\emph{Quantum Inf. Process.}}   \textbf{22}~(1)
	(\bibinfo{year}{2023}) \bibinfo{pages}{Paper No. 41}.
	%Type = Article
	\bibitem[{Connelly et~al.(2017)Connelly, Grammel, Kraut, Serazo, and
			Tamon}]{CGKST17}
	\bibinfo{author}{E.~Connelly}, \bibinfo{author}{N.~Grammel},
	\bibinfo{author}{M.~Kraut}, \bibinfo{author}{L.~Serazo},
	\bibinfo{author}{C.~Tamon},
	\newblock \bibinfo{title}{Universality in perfect state transfer},
	\newblock \bibinfo{journal}{\emph{Linear Algebra Appl.}}   \textbf{531}
	(\bibinfo{year}{2017}) \bibinfo{pages}{516--532}.
	%Type = Book
	\bibitem[{Coutinho and Godsil(2021)}]{CG21b}
	\bibinfo{author}{G.~Coutinho}, \bibinfo{author}{C.~Godsil},
	\bibinfo{title}{Graph spectra and continuous quantum walks},
	\bibinfo{publisher}{In preparation}, \bibinfo{year}{2021}.
	%Type = Book
	\bibitem[{Godsil(1993)}]{G93}
	\bibinfo{author}{C.~Godsil}, \bibinfo{title}{Algebraic combinatorics},
	\bibinfo{publisher}{New York: Chapman \& Hall}, \bibinfo{year}{1993}.
	%Type = Article
	\bibitem[{Godsil(2011)}]{Go08}
	\bibinfo{author}{C.~Godsil},
	\newblock \bibinfo{title}{Periodic graphs},
	\newblock \bibinfo{journal}{\emph{Electron. J. Combin.}}   \textbf{18}~(1)
	(\bibinfo{year}{2011}) \bibinfo{pages}{\#P23}.
	%Type = Article
	\bibitem[{Godsil(2012{\natexlab{b}})}]{Go12b}
	\bibinfo{author}{C.~Godsil},
	\newblock \bibinfo{title}{State transfer on graphs},
	\newblock \bibinfo{journal}{\emph{Discrete Math.}}   \textbf{312}~(1)
	(\bibinfo{year}{2012}{\natexlab{b}}) \bibinfo{pages}{129--147}.
	%Type = Article
	\bibitem[{Godsil(2012{\natexlab{a}})}]{Go12a}
	\bibinfo{author}{C.~Godsil},
	\newblock \bibinfo{title}{When can perfect state transfer occur?},
	\newblock \bibinfo{journal}{\emph{Electron. J. Linear Algebra}}   \textbf{23}
	(\bibinfo{year}{2012}{\natexlab{a}}) \bibinfo{pages}{877--890}.
	%Type = Book
	\bibitem[{Godsil and Royle(2001)}]{GR01}
	\bibinfo{author}{C.~Godsil}, \bibinfo{author}{G.~Royle},
	\bibinfo{title}{Algebraic graph theory}, volume \bibinfo{volume}{207} of
	\emph{\bibinfo{series}{Graduate Texts in Mathematics}},
	\bibinfo{publisher}{Springer-Verlag, New York}, \bibinfo{year}{2001}.
	%Type = Article
	\bibitem[{Guo and Mohar(2017)}]{GM17}
	\bibinfo{author}{K.~Guo}, \bibinfo{author}{B.~Mohar},
	\newblock \bibinfo{title}{Hermitian adjacency matrix of digraphs and mixed
		graphs},
	\newblock \bibinfo{journal}{\emph{J. Graph Theory}}   \textbf{85}~(1)
	(\bibinfo{year}{2017}) \bibinfo{pages}{217--248}.
	%Type = Article
	\bibitem[{Kadyan and Bhattacharjya(2023)}]{MB21a}
	\bibinfo{author}{M.~Kadyan}, \bibinfo{author}{B.~Bhattacharjya},
	\newblock \bibinfo{title}{Integral mixed circulant graphs},
	\newblock \bibinfo{journal}{\emph{Discrete Math.}}   \textbf{346}~(1)
	(\bibinfo{year}{2023}) \bibinfo{pages}{113142}.
	%Type = Article
	\bibitem[{Liu and Li(2015)}]{LL15}
	\bibinfo{author}{J.~Liu}, \bibinfo{author}{X.~Li},
	\newblock \bibinfo{title}{Hermitian-adjacency matrices and {H}ermitian energies
		of mixed graphs},
	\newblock \bibinfo{journal}{\emph{Linear Algebra Appl.}}   \textbf{466}
	(\bibinfo{year}{2015}) \bibinfo{pages}{182--207}.
	%Type = Article
	\bibitem[{Sett et~al.(2019)Sett, Pan, Falloon, and Wang}]{SPFW19}
	\bibinfo{author}{A.~Sett}, \bibinfo{author}{H.~Pan}, \bibinfo{author}{P.~E.
		Falloon}, \bibinfo{author}{J.~B. Wang},
	\newblock \bibinfo{title}{Zero transfer in continuous-time quantum walks},
	\newblock \bibinfo{journal}{\emph{Quantum Inf. Process.}}   \textbf{18}~(5)
	(\bibinfo{year}{2019}) \bibinfo{pages}{Paper No. 159, 18}.
	%Type = Article
	\bibitem[{Song(2024)}]{S22}
	\bibinfo{author}{X.~Song},
	\newblock \bibinfo{title}{Quantum state transfer on integral oriented circulant
		graphs},
	\newblock \bibinfo{journal}{\emph{Appl. Math. Comput.}}   \textbf{464}
	(\bibinfo{year}{2024}) \bibinfo{pages}{Paper No. 128391}.
	%Type = Article
	\bibitem[{Song and Lin(2024)}]{SL22}
	\bibinfo{author}{X.~Song}, \bibinfo{author}{H.~Lin},
	\newblock \bibinfo{title}{State transfer on integral mixed circulant graphs},
	\newblock \bibinfo{journal}{\emph{Discrete Math.}}   \textbf{347}~(1)
	(\bibinfo{year}{2024}) \bibinfo{pages}{Paper No. 113727, 13}.
	%Type = Manual
	\bibitem[{{The Sage Developers}(2023)}]{sage}
	\bibinfo{author}{{The Sage Developers}}, \bibinfo{title}{{S}age {M}athematics
			{S}oftware ({V}ersion 10.2)}, \bibinfo{organization}{The Sage Development
		Team}, \bibinfo{year}{2023}. \bibinfo{note}{{\tt https://www.sagemath.org}}.

\end{thebibliography}
\end{document}